\documentclass[A4]{article}

\usepackage{amsmath,amsfonts,amsthm,amssymb,amscd,color}
\usepackage[dvipdfmx]{graphicx}
\usepackage[format=hang,font=footnotesize]{caption}

\theoremstyle{plain} \newtheorem{theorem}{Theorem}[section]
\newtheorem{lemma}[theorem]{Lemma}

 \theoremstyle{definition}
\newtheorem{definition}[theorem]{Definition} \theoremstyle{remark}

\newcommand{\C}{\mathbb{C}}

\def\({\left(}
\def\){\right)}
\def\<{\left\langle}
\def\>{\right\rangle}
\numberwithin{equation}{section}

\begin{document}

\title{Weak limit theorem for a nonlinear quantum walk}

\author {Masaya Maeda, Hironobu Sasaki, Etsuo Segawa, Akito Suzuki, Kanako Suzuki}

\maketitle

\begin{abstract}
This paper continues the study of large time behavior 
of a nonlinear quantum walk begun in \cite{MSSSS}. 
In this paper, we provide a weak limit theorem 
for the distribution of the nonlinear quantum walk. 
The proof is based on the scattering theory of 
the nonlinear quantum walk and 
the limit distribution is obtained in terms of its asymptotic state. 
\end{abstract}

\section{Introduction}

This paper continues the study of a one-dimensional 
nonlinear quantum walk (NLQW)
begun in \cite{MSSSS}, where we developed a scattering theory 
for NLQW. 
The model treated there %in \cite{MSSSS} 
covers  a nonlinear optical Galton board \cite{NPR07PRA, MDB15PRE},
a quantum walk with a feed-forward quantum coin \cite{SWH14SR}, 
nonlinear discrete dynamics %exhibiting a solitonic behavior
\cite{LKN15PRA},
and a model exhibiting topological phenomena \cite{GTB16PRA}.  
For more details on earlier works, 
we refer to the previous paper\cite{MSSSS}. 
In a forthcoming companion paper \cite{MSSSSnum}, 
we numerically study a solitonic behavior of NLQW. 
In this paper, we study a weak limit theorem (WLT) for NLQW.  
The WLT for the one-dimensional (linear) 
quantum walk (QW)
was first found by Konno \cite{Kon02}, proved in \cite{Kon05},
and then generalized by several authors 
\cite{EEKST16, EKOS17JPA, FFS18, GJS04PRE, HiKoSaSe14, HiSe17,
KLS13, MaSe15, RiSuTi17b, Suzuki16QIP}.  
The WLT states that
\begin{center}
$\displaystyle \frac{X_t}{t}$ converges in law to a random variable $V$
as $t \to \infty$,
\end{center}
where $X_t$ is a random variable denoting the position of a quantum walker at time 
$t = 0, 1, 2, \ldots$. 
Because $X_t/t$ is the average velocity of the walker,
$V$ is interpreted as the asymptotic velocity of the walker and 
hence WLT well describes the asymptotic behavior of the walker. 
Here, the probability distribution of $X_t$ 
is naturally defined according to Born's rule as
\[ P(X_t = x) = \|\Psi_t(x)\|_{\mathbb{C}^2}^2,
	\quad x \in \mathbb{Z}, \]
where 
$\Psi_t$ is the state of the walker at time $t$,
which is in the state space 
 $\mathcal{H} :=  l^2(\mathbb{Z};\mathbb{C}^2)$. 
The state evolution is governed by 
\[ \Psi_{t+1}(x) = P(x+1) \Psi_t(x+1) + Q(x-1) \Psi_t(x-1),
	\quad x \in \mathbb{Z}, \ t = 0, 1, 2, \ldots,  \] 
where $P(x)$ and $Q(x) \in M(2; \mathbb{C})$ satisfy
$P(x) + Q(x) = : C(x) \in U(2)$. 
More precisely, the state at time $t$ is given by 
$\Psi_t = U_{\rm L}^t \Psi_0$,
where $\Psi_0$ is the initial state,
which is a normalized vector in $\mathcal{H}$,
%with $\|\Psi_0\|_\mathcal{H}=1$
and $U_{\rm L}$ is the evolution operator defined as follows.  
Let $\hat C$ be the coin operator defined as the multiplication by $C(x)$ and 
$S$ be the shift operator, 
i.e., $(\hat C\Psi)(x) := C(x) \Psi(x)$ and 
$(S\Psi)(x) := \ ^{\rm t}(\Psi_1(x+1), \Psi_2(x-1))$
($x \in \mathbb{Z}$)
for $\Psi = \ ^{\rm t}(\Psi_1, \Psi_2) \in \mathcal{H}$. 
The evolution operator $U_{\rm L}$ 
is a unitary operator defined as $U_{\rm L}=S \hat C$. 
As shown in \cite{Suzuki16QIP},
if $C(x) = C_0 + O(|x|^{-1-\epsilon})$ with some $C_0 \in U(2)$
and $\epsilon > 0$ independent of $x$,
then WLT is proved and the limit distribution $\mu_V$ is 
expressed  in terms of the wave operator  
$W_+ := \mbox{s-}\lim_{t \to \infty}
	 U_{\rm L}^{-t}U_0^t \Pi_{\rm ac}(U_0)$,
where $U_0 = S \hat C_0$ and 
$\Pi_{\rm ac}(U_0)$ is the projection onto the subspace of 
absolute continuity. 
See also \cite{EEKST16, RiSuTi17a, RiSuTi17b} for anisotropic cases. 

In the case of NLQW, 
the dynamics is governed by 
\begin{equation}
\label{NLev} 
u(t+1, x)
	= (\hat Pu(t))(x+1) + (\hat Qu(t))(x-1),
	\quad x \in \mathbb{Z}, \ t = 0, 1, 2, \ldots,  
\end{equation} 
where 
$t \mapsto u(t) := u(t, \cdot) \in \mathcal{H}$ is 
in $l^\infty(\mathbb{N} \cup \{0\}; \mathcal{H})$. 
$\hat P$ and $\hat Q$ are nonlinear maps on $\mathcal{H}$
and give a norm preserving nonlinear map $\hat C:\mathcal{H}
\ni u \mapsto \hat Cu := \hat Pu + \hat Qu$.
Although the dynamics is similar to the linear quantum walk,
it does not define a quantum system. 
However, %as pointed out in \cite{MSSSS},
%if $u(0,\cdot) \in \mathcal{H}$ is normalized,  
%the position $X_t$ of a nonlinear quantum walker 
%can be defined as 
\begin{equation}
\label{pdX} 
%P(X_t = x) =
p_t(x) :=  \|u(t,x)\|_{\mathbb{C}^2}^2,
	\quad x \in \mathbb{Z}. 
\end{equation}
defines a probability distribution. 
Indeed, 
similarly to the linear quantum walk,
the dynamics \eqref{NLev} is expressed as
$u(t+1, \cdot) = U u(t, \cdot)$, 
$t = 0, 1, 2, \ldots$,  
where 
$U:=S\hat C$ is a nonlinear map on $\mathcal{H}$. 
%We then define the nonlinear evolution operator 
%$U(t):\mathcal{H} \to 
%\ell^\infty(\mathbb{N} \cup \{0\};\mathcal{H})$ as
%\[ U(t)u_0 = u(t, \cdot), \quad t = 0, 1, 2, \ldots. \]
%Observe that $u_0 \mapsto u = U(\cdot)u_0$ defines a nonlinear map
%from $\mathcal{H}$ to $\ell^\infty_t(\mathbb{N};\mathcal{H})$. 
Because $U$ preserves the norm, 
\eqref{pdX} defines the probability distribution
provided that the initial state  $u(0, \cdot) = u_0 \in \mathcal{H}$ 
is a normalized vector. 
We use $X_t$ to denote
the random variable that follows \eqref{pdX},
{\it i.e.}, $P(X_t = x) = p_t(x)$.   
Of course $X_t$ never describes the position of a walker
that occupies any single position in $\mathbb{Z}$,  
but we dare to call $X_t$ the position of a nonlinear quantum walker
in analogy with the linear quantum walk. 
It is mathematically more convenient 
to study the limit behavior of $X_t$ than the distribution $p_t(x)$ itself.

In this paper, we consider a nonlinear coin given by
\[ (\hat C u)(x) 
	= C_{\rm N}(g|u_1(x)|^2,g|u_2(x)|^2) u(x), 
	\quad x \in \mathbb{Z}
	\quad \mbox{for $u = \ ^{\rm t}(u_1, u_2) \in \mathcal{H}$}, \]
where $g > 0$ controls the strength of the nonlinearity and
$C_{\rm N}:[0,\infty) \times [0,\infty) \ni (s_1, s_2)
	\mapsto C_{\rm N}(s_1, s_2) \in U(2)$ 
with $C_{\rm N}(0,0) =: C_0 \in U(2)$.  
As was shown in \cite{MSSSS}, 
in the weak nonlinear regime,
%and the (1,1)-entry $a$ of  $C_0$ satisfies $0 < |a| < 1$, 
$U(t)u_0$ scatters, {\it i.e.},  
$\lim_{t \to \infty} \|u(t, \cdot) - U_0^t u_+\|_{\mathcal{H}} = 0$
with some asymptotic state $u_+ \in \mathcal{H}$. 
The aim of this paper is to establish 
WLT for $X_t$ that follows \eqref{pdX}
and prove that the limit distribution is given by
\[ \mu_V(dv) = w(v) f_{\rm K}(v;|a|) dv, \]
where $w(v)$ is a function expressed in terms of $u_+$ and 
$f_{\rm K}(v;r)$ is the Konno function ($r>0$).

The rest of this paper is organized as follows. 
Sec. 2 is devoted to reviewing the results of \cite{MSSSS}. 
We state our main results in Sec. 3
and give proofs in Sec. 4. 
 
%=============================%
\section{Preliminaries}
%=============================%
In this section, we review the definition of NLQW 
and results obtained in \cite{MSSSS}. 
Throughout this paper, we set 
$\mathcal{H} = l^2(\mathbb{Z};\mathbb{C}^2)$
%use $\|\cdot\|_\mathcal{H}$ 
%and $\langle \cdot, \cdot \rangle_\mathcal{H}$
%to denote the norm and inner product of $\mathcal{H}$,
and drop the subscript $\mathcal{H}$ in the norm and inner product
when there is no ambiguity.
Let  
\[ C_{\rm N}:[0, \infty) \times [0,\infty) \rightarrow U(2) \]
satisfy
\begin{align*}
C_0 := C_{\rm N}(0,0) =\begin{pmatrix} a & b \\-\bar b & \bar a \end{pmatrix}\quad \mbox{with $|a|^2+|b|^2=1$ and $0<|a|<1$.} 
\end{align*}
We define a nonlinear coin operator $\hat C$ as
\begin{align} 
\label{coin} 
(\hat C u)(x) = C_{\rm N}(g |u_1(x)|^2, g |u_2(x)|^2) u(x),
\quad x \in \mathbb{Z}
	\quad \mbox{for $u = \ ^{\rm t}(u_1,u_2) \in \mathcal{H}$},
\end{align}
where $g>0$ is a constant that controls the strength of the nonlinearity. 
Let $u_0 \in \mathcal{H}$ be the initial state of a walker 
with $\|u_0\|=1$. 
The state $u(t)$ of the walker at time $t =1, 2, \ldots$
is defined by induction as follows. 	
\[ u(0) = u_0, \quad u(t+1) = U u(t), \quad t = 0, 1,2, \ldots,  \]
where $U = S \hat C$. 
We then define a nonlinear evolution operator $U(t)$ as 
\[ U(t) u_0 = u(t). \]	
Similarly, we define a linear coin operator $\hat C_0$ as
$(\hat C_0 u)(x) = C_0 u(x)$
and set $U_0=S \hat C_0$. 
By scattering, we mean the following:

\begin{definition}\label{def:scatters}
We say $U(t)u_0$ scatters if there exists $u_+\in \mathcal H$ such that
\[ \lim_{t \to \infty} \|U(t)u_0 - U_0^t u_+\|=  0. \]
\end{definition}

We use $U_{g=1}(t)$ to denote the evolution operator $U(t)$
that has the nonlinear coin $\hat C$ defined in \eqref{coin} 
with $g=1$. 
As mentioned in the previous paper \cite{MSSSS}, 
the smallness of $\|u_0\|_\mathcal{H}$ and $\|u_0\|_{l^1}$ 
corresponds to the smallness of $g$,
because 
\[ U(t)  u_0 = \frac{1}{\sqrt{g}} U_{g=1}(t) v_0
	\quad \mbox{with $v_0 := \sqrt{g} u_0$}.  \]
Thus, the result in \cite{MSSSS} is reformulated as follows. 
We use $\|A\|_{\C^2\to \C^2}$ tot denote the operator norm of the matrix $A$,
{\it i.e.}, $\|A\|_{\C^2\to \C^2}:=\sup_{v\in \C^2, \|v\|_{\C^2}=1}\|Av\|_{\C^2}$.

\begin{theorem}[\cite{MSSSS}]\label{thm:scat}
Assume that $C_{\rm N} \in C^1(\Omega;U(2))$ with some domain 
$\Omega$ including $[0,\infty) \times [0,\infty)$ and there exists 
$c_0>0$ and $m \geq 2$ such that 
\ $\|C_{\rm N}(s_1,s_2)-C_0\|_{\C^2\to \C^2}\leq c_0  (s_1+s_2)^{m}$ and  $\|\partial_{s_j}C_{\rm N}(s_1,s_2)\|_{\C^2\to \C^2}\leq c_0  (s_1+s_2)^{m-1}$ for $j=1$ or $2$. 
Let $u_0 \in \mathcal{H}$ be a normalized vector. 
Suppose in addition that 
either of the following conditions holds: (1) $m \geq 3$; 
(2) $m=2$ and $u_0 \in l^1(\mathbb{Z}, \mathbb{C}^2)$. 
Then $U(t)u_0$ scatters if $g$ is sufficiently small. 
\end{theorem}

%======================================
\section{Weak limit theorem}
%======================================
Our aim is to establish the weak limit theorem for 
the position $X_t$ of a walker at time $t$ that follows 
the probability distribution
\begin{equation}
\label{Xt_dist} 
P(X_t = x) = \|u(t,x)\|_{\mathbb{C}^2}^2,  \quad x \in \mathbb{Z}, 
\end{equation} 
where $u(t, \cdot) := U(t)u_0$ with $\|u_0\|_\mathcal{H}=1$.
By Theorem \ref{thm:scat}, if $g$ is sufficiently small,
then $U(t)u_0$ scatters, {\it i.e.}, there exists $u_+ \in \mathcal{H}$
such that  
\[ 
\lim_{t \to \infty} \|U(t)u_0 - U_0^t u_+\|= 0.
\]   
Let $\hat v_0$ be the asymptotic velocity operator
for $U_0 = S \hat C_0$, which is a unique self-adjoint operator
such that 
\begin{equation}
\label{vlc} 
e^{i \xi \hat v_0} 
= \mbox{s-}\lim_{t \to +\infty} e^{i \xi \hat x_0(t)/t}. 
\end{equation} 
Here $\hat x_0(t) = U_0^{-t} \hat x U_0$ is 
the Heisenberg operator of the position operator $\hat x$. 
See \cite{GJS04PRE,Suzuki16QIP} for more details. 
We give the precise definition of $\hat v_0$ in \eqref{def:vlc}. 
We use $E_{A}(\cdot)$ to denote the spectral projection of 
a self-adjoint operator $A$. 
%===============
\begin{theorem}[weak limit theorem]
\label{p_wlt}
Let $X_t$, $\hat v_0$, and $u_+$ be as above. 
Then there exists a random variable $V$
such that $X_t/t$ converges in law to $V$,
whose distribution $\mu_V$ is given by
\[ \mu_V(dv) =  d \|E_{\hat v_0} (v) u_+\|^2. \]
%w(v) f_{\rm K}(v; |a|) dv. 
%Here $w \in L^1([-|a|, |a|]; f_{\rm K}(v;|a|) dv)$ 
%is a unique function determined by $u_+$ and $C_0$. 
\end{theorem}
%===================
In what follows, 
we provide an explicit formula for the density function of $\mu_V$ obtained in Theorem \ref{p_wlt}. 
To this end, we proceed along the lines of \cite{RiSuTi17b}. 
Let $f_{\rm K}$ be the Konno function defined for all $r > 0$ as
\[ f_{\rm K}(v;r) = \begin{cases} 
	\frac{\sqrt{1-r^2}}{\pi(1-v^2)\sqrt{r^2 -v^2}},
		& |v| < r, \\
		0, & |v| \geq 0. 
		\end{cases} \]
Similarly to \cite{RiSuTi17b}, we introduce operators 
\[ K_{j,m}: \mathcal{H}
\to \mathcal{G}
	:= L^2([-|a|, |a|], f_{\rm K}(v;|a|)dv/2),
	\quad
j=1,2, \ m=0,1 \]
as follows. 
Let $\mathbb{T}:= \mathbb{R}/2\pi \mathbb{Z}$. 
Because $U_0$ is translation invariant, 
it can be decomposed by the Fourier transformation 
$F: \mathcal{H} 
\to L^2(\mathbb{T}; \mathbb{C}^2; dk/2\pi)$
and the Fourier transform $F U_0 F^{-1}$ is the multiplication operator by
%$\hat U_0(k) \in U(2)$ 
\[ \hat U_0(k) 
= \begin{pmatrix} 
	e^{ik} a & e^{ik}b \\
	- e^{-ik} \bar b & e^{-ik} \bar a
	 \end{pmatrix} \in U(2), 
	 \quad k \in \mathbb{T}.  \]
We use $\varphi_j(k)$ to denote the normalized eigenvectors
of $\hat U_0(k)$ corresponding to the eigenvalues 
\[ \lambda_j(k) 
	= |a| \cos (k + \theta_a) 
		+ i (-1)^{j-1} \sqrt{|b|^2 + |a|^2 \sin (k+\theta_a)},
			\quad j = 1, 2.  \] 
Let 
$k_{k,m}: [-|a|, |a|] \to I_m := [\pi(m-1/2) - \theta_a, \pi(m+1/2) -\theta_a]$
be a function defined as
\[ 
k_{j,m}(v) 
= - \theta_a + m \pi 
	+\arcsin \left( \frac{(-1)^{j+m} |b| v}{|a|\sqrt{1-v^2}} \right),
\quad
j=1,2, \ m=0,1,
\]
where $\theta_a \in [0,2\pi)$ is the argument of $a$. 
By direct calculation, $k_{j,m}$ is differentiable in $(-|a|, |a|)$
and
\[ \frac{d}{dv} k_{j,m} = (-1)^{j+m} \pi f_{\rm K}(v, |a|). \]
We now define the operators $K_{j,m}$ as
\[ (K_{j,m}u)(v) 
	= \langle \varphi_j(k_{j,m}(v)), 
		\hat u (k_{j,m}(v)) \rangle_{\mathbb{C}^2},
	\quad v \in [-|a|, |a|], \]
where $\hat u$ is the Fourier transform of $u \in \mathcal{H}$. 

\begin{theorem}
\label{thm:dist}
Let $u_+$ and $V$ be  as Theorem \ref{p_wlt}.
Then 
\[ \mu_V(dv) = w(v) f_{\rm K}(v,|a|) dv, \]
where
\[ w(v) = \frac{1}{2} 
	\sum_{j=1,2} \sum_{m=0,1} 
		|(K_{j,m} u_+)(v)|^2,
		\quad v \in [-|a|, |a|].   \]
\end{theorem}

%======================
\section{Proofs of Theorems}
%======================

The proofs of Theorem \ref{p_wlt} and \ref{thm:dist} proceed along the same lines
as that of \cite{Suzuki16QIP}[Corollary 2.4]. 
We suppose that $\|u_0\|_\mathcal{H}=1$. 
Let $\hat x$ be the position operator. 
By \eqref{Xt_dist}, the characteristic function of $X_t/t$
is given by
\begin{equation}
\label{chrct_X_t} 
\mathbb{E}\left(e^{i \xi X_t/t} \right)  
	= \left\langle U(t) u_0, 
		e^{i \xi \hat x/t} U(t) u_0 \right\rangle,
		\quad \xi \in \mathbb{R}, 
\end{equation}
where $\mathbb{E}(X)$ denotes the expectation value of 
a random variable $X$.  
%Let $\hat v_0$ be the asymptotic velocity operator
%for $U_0 = S \hat C_0$, which is a unique self-adjoint operator
%such that 
%\begin{equation}
%\label{vlc} 
%e^{i \xi \hat v_0} 
%= \mbox{s-}\lim_{t \to +\infty} e^{i \xi \hat x_0(t)}. 
%\end{equation} 
%Here $\hat x_0(t) = U_0^{-t} \hat x U_0$ is 
%the Heisenberg operator of $\hat x$
%(see \cite{GJS04PRE,Suzuki16QIP} for more details). 
The asymptotic velocity operator $\hat v_0$ in \eqref{vlc}
is defined via the Fourier transform: 
$F v_0 F^{-1}$ is the multiplication operator by 
\begin{equation}
\label{def:vlc} \hat v_0(k) 
	= \sum_{j=1,2} v_j(k) |\varphi_j(k) \rangle \langle \varphi_j(k)|,
	\quad k \in \mathbb{T}, 
\end{equation}
where 
\[ v_j(k) := \frac{i}{\lambda_j(k)} \frac{d}{dk} \lambda_j(k) 
	= \frac{(-1)^j |a| \sin (k + \theta_a)} 
		{\sqrt{|b|^2 + \sin^2(k+\theta_a)}}. \]
As was shown in \cite{RiSuTi17b},
$v_j: I_m%[\pi(m-1/2) - \theta_a, \pi(m+1/2) -\theta_a] 
\to [-|a|, |a|]$ 
is the inverse function of $k_{j,m}$.

%%%%%%%%%%%%%
\begin{lemma}
\label{lem_A1}
\[ \lim_{t \to \infty}
	\langle U(t) u_0, e^{i \xi \hat x/t} U(t) u_0 \rangle
	= \langle u_+, e^{i \xi \hat v_0} u_+ \rangle \]
\end{lemma}
%%%%%%%%%%%
\begin{proof}
A direct calculation yields
\begin{align*}
& |\langle U(t) u_0, e^{i \xi \hat x/t} U(t) u_0 \rangle
	- \langle u_+, e^{i \xi \hat v_0} u_+ \rangle| \\
& \quad \leq |\langle U(t) u_0 - U_0^t u_0,
		 e^{i \xi \hat x/t} U(t) u_0 \rangle| 
	+ |\langle U_0^t u_0, 
		e^{i \xi \hat x/t}( U(t) u_0 - U_0^t u_0) \rangle| \\
& \qquad + |\langle U_0^t u_0, e^{i \xi \hat x/t}  U_0^t u_0\rangle
		- \langle u_+, e^{i \xi \hat v_0} u_+ \rangle  \rangle| \\
& \quad =: I_1(t) + I_2(t) + I_3(t).
\end{align*}
Because $e^{i\xi \hat x/t}$ and $U(t)$ preserve the norm
and $U(t)u_0$ scatters, 
$\lim_{t \to \infty} I_1(t)= \lim_{t \to \infty} I_2(t)=0$. 
By \eqref{vlc}, $\lim_{t \to \infty}I_3(t)=0$. 
Hence the proof is completed. 
\end{proof}
%%%%%%%%%%%%
\begin{proof}[Proof of Theorem \ref{p_wlt}]
By \eqref{chrct_X_t} and Lemma \ref{lem_A1},
\begin{equation}
\label{lim_chrct}
\lim_{t \to \infty} E(e^{i \xi X_t/t})
	= \langle u_+, e^{i \xi \hat v_0} u_+ \rangle
	= \int_{[-|a|, |a|]} e^{i \xi v} d \|E_{\hat v_0}(v) u_+\|^2,
\end{equation}
where we have used the spectral theorem. 
The right-hand side in the above equation is equal to 
the characteristic function of a random variable $V$
following the probability distribution 
$\mu_V =  \|E_{\hat v_0}(\cdot) u_+\|^2$. 
This completes the proof of Theorem \ref{p_wlt}. 
\end{proof}

In what follows, we prove Theorem \ref{thm:dist}. 
The following lemma is proved similarly to \cite{RiSuTi17b}. 
%===============
\begin{lemma}
\label{newex}
We use $\hat G$ to denote the multiplication operator on $\mathcal{G}$ 
by a Borel function $G:[-|a|,|a|] \to \mathbb{C}$. 
Then
\[ G(\hat v_0) = \sum_{j=1,2} \sum_{m=0,1}
	K_{j,m}^* \hat G K_{j,m}. \]
\end{lemma} 
%================
\begin{proof}[Proof of Theorem \ref{thm:dist}]
It suffices to prove %\eqref{wldist}.
\begin{equation}
\label{wldist} 
\langle u_+, e^{i \xi \hat v_0} u_+ \rangle 
= \int_{[-|a|, |a|]} e^{i \xi v} w(v) f_K(v;|a|)dv,
	\quad \xi \in \mathbb{R}. 
\end{equation}
Let $G(v) = e^{i \xi v}$. 
By Lemma \ref{newex}, the left-hand side of \eqref{wldist} is
\begin{align*} 
\langle u_+, e^{i \xi \hat v_0} u_+ \rangle 
& = \sum_{j=1,2} \sum_{m=0,1}
	 \left\langle K_{j,m}u_+, 
	 	\hat G K_{j,m} u_+ \right\rangle_{\mathcal{G}}. 
\end{align*}
Because 
\[ \left\langle K_{j,m}u_+, 
	 	\hat G K_{j,m} u_+ \right\rangle_{\mathcal{G}} 
= \int_{[-|a|, |a|]} e^{i \xi v}
	| (K_{j,m} u_+)(v) |^2
	 f_{\rm K}(v;|a|) dv/2,  \]
the proof of theorem is complete. 
\end{proof}

%==================================
\section*{Acknowledgments}  
M.M. was supported by the JSPS KAKENHI Grant Numbers JP15K17568, JP17H02851 and JP17H02853.
H.S. was supported by JSPS KAKENHI Grant Number JP17K05311.
E.S. acknowledges financial support from 
the Grant-in-Aid for Young Scientists (B) and of Scientific Research (B) Japan Society for the Promotion of Science (Grant No.~16K17637, No.~16K03939).
A. S. was supported by JSPS KAKENHI Grant Number JP26800054. 
K.S acknowledges JSPS the Grant-in-Aid for Scientific Research (C) 26400156.
%===================================== 

\medskip

Masaya Maeda, Hironobu Sasaki

Department of Mathematics and Informatics,
Faculty of Science,
Chiba University,
Chiba 263-8522, Japan

{\it E-mail Address}: {\tt maeda@math.s.chiba-u.ac.jp, sasaki@math.s.chiba-u.ac.jp}

\medskip

Etsuo Segawa

Graduate School of Information Sciences, 
Tohoku University,
Sendai 980-8579, Japan

{\it E-mail Address}: {\tt e-segawa@m.tohoku.ac.jp}

\medskip

Akito Suzuki

Division of Mathematics and Physics,
Faculty of Engineering,
Shinshu University,
Nagano 380-8553, Japan

{\it E-mail Address}: {\tt akito@shinshu-u.ac.jp}

\medskip

Kanako Suzuki

College of Science, Ibaraki University,
2-1-1 Bunkyo, Mito 310-8512, Japan

{\it E-mail Address}: {\tt kanako.suzuki.sci2@vc.ibaraki.ac.jp}

\end{document}